\documentclass[letterpaper, 10 pt, conference]{ieeeconf}  

\IEEEoverridecommandlockouts    

\usepackage[T1]{fontenc}
\usepackage{xcolor, times}
\usepackage{multicol}
\definecolor{mycitecolor}{RGB}{71, 191, 38}
\definecolor{mylinkcolor}{RGB}{40, 115, 201}

\usepackage{mathtools} 
\usepackage{amsmath} 
\usepackage{amssymb} 
\usepackage{fancyhdr}
\usepackage{dsfont}
\let\labelindent\relax
\usepackage{enumitem}

\usepackage{amsthm}
\usepackage{subcaption}
\usepackage[font=footnotesize]{caption}
\usepackage{diagbox}
\usepackage{booktabs}
\usepackage{arydshln}
\usepackage{tikz}
\usetikzlibrary{patterns,angles,quotes,arrows}
\usetikzlibrary{positioning}

\usepackage{algorithm}
\usepackage{algpseudocode}

\newtheoremstyle{mystyle}
  {}
  {}
  {}
  {}
  {\bfseries}
  {.}
  { }
  {}
\newtheorem{remark}{Remark}
\newtheorem{theorem}{Theorem}

\newtheorem{lemma}{Lemma}
\newtheorem{definition}{Definition}
\newtheorem{problem}{Problem}

\newtheorem{assumption}{Assumption}

\makeatletter
\let\NAT@parse\undefined
\makeatother
\usepackage[bookmarks=true, colorlinks, citecolor=mycitecolor,linkcolor=mylinkcolor,urlcolor=mycitecolor]{hyperref}  

\title{\LARGE \bf
Learning to Provably Satisfy High Relative Degree Constraints for Black-Box Systems*
}

\author{Jean-Baptiste Bouvier, Kartik Nagpal and Negar Mehr
\thanks{*This work is supported by the National Science Foundation, under grants ECCS-2145134 CAREER Award, CNS-2218759, and CCF-2211542.}
\thanks{Jean-Baptiste Bouvier, Kartik Nagpal and Negar Mehr are with the Department of Mechanical Engineering, University of California Berkeley, Berkeley, CA 94709, USA {\tt\small \{bouvier3, kartiknagpal, negar\}@berkeley.edu}}%
}

\begin{document}

\maketitle
\thispagestyle{empty}
\pagestyle{empty}

\begin{abstract}
    In this paper, we develop a method for learning a control policy guaranteed to satisfy an affine state constraint of high relative degree in closed loop with a black-box system. Previous reinforcement learning (RL) approaches to satisfy safety constraints either require access to the system model, or assume control affine dynamics, or only discourage violations with reward shaping. Only recently have these issues been addressed with POLICEd RL, which guarantees constraint satisfaction for black-box systems. However, this previous work can only enforce constraints of relative degree 1. To address this gap, we build a novel RL algorithm explicitly designed to enforce an affine state constraint of high relative degree in closed loop with a black-box control system. Our key insight is to make the learned policy be affine around the unsafe set and to use this affine region to dissipate the inertia of the high relative degree constraint. We prove that such policies guarantee constraint satisfaction for deterministic systems while being agnostic to the choice of the RL training algorithm. Our results demonstrate the capacity of our approach to enforce hard constraints in the Gym inverted pendulum and on a space shuttle landing simulation.
\end{abstract}

\section{Introduction}

The lack of safety guarantees in reinforcement learning (RL) has been impeding its wide deployment in real-world settings~\cite{RL_challenges}. Safety in RL is traditionally captured by state constraints preventing the system from entering unsafe regions \cite{review_safety_levels}. 
This issue has been investigated by numerous approaches and most commonly under the framework of constrained Markov decision processes (CMDPs)~\cite{CMDP, state_wise_constrained_MDP, constrained_policy_optimization}. CMDPs only \emph{encourage} policies to respect safety constraints by penalizing the expected violations \cite{RL_soft_constraint}; however, they do not provide any satisfaction guarantees~\cite{review_safe_RL}.
For safety-critical tasks, such as autonomous driving or human-robot interactions, safety guarantees are primordial and require the learned policy to maintain constraint respect.

A few RL attempts at learning provably safe policies have involved control barrier functions (CBFs) \cite{BarrierNet}, backward reachable sets \cite{backward_reachable}, and projection of control inputs onto safe sets \cite{safe_exploration, Optlayer}. However, all these methods require precise knowledge of the system dynamics, which is usually not available in RL.
To circumvent this issue and study black-box systems without an analytical model of the dynamics, the common approach has been to learn safety certificates \cite{sablas, ma2022joint, yang2023model}. Yet, by learning to approximate CBFs, the formal safety guarantees of these methods hinge upon the quality of their CBF approximation.
More reliable safety guarantees have been established by recent work \cite{POLICEd_RL}, whose POLICEd RL approach designs a repulsive buffer to enforce constraint satisfaction in closed-loop with a black-box system. However, work \cite{POLICEd_RL} along with most other safe RL works such as \cite{safe_exploration, Optlayer, sablas, yang2023model} are limited to constraints of relative degree $1$.
In contrast, our approach enforces inviolable \emph{constraints of high relative degree} in closed-loop with a learned control policy while exclusively using a \emph{black-box model} of the system dynamics.

The \emph{relative degree of a constraint} describes how many times a constraint needs to be differentiated before a control input appears in its expression. The higher the relative degree, the more inertia the constraint has and the more challenging its satisfaction is \cite{ZCBF}. Inspired by the recent extensions of CBFs to high relative degrees \cite{Exponential_CBF, HOCBF, ZCBF}, we propose a backstepping inspired approach~\cite{backstepping}, which is compatible with systems non-affine in control and with black-box systems contrary to these CBF methods.

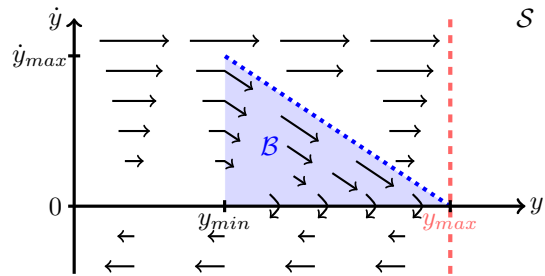
\begin{figure}[t!]
    \centering
    \begin{tikzpicture}[scale=1]
        \def\ymin{2}    
        \def\ymax{5}    
        \def\ydotmax{2} 
        \def\tick{0.08} 
        \def\width{6} 
        \def\height{\ydotmax + 0.5} 
        \def\arrowLen{0.6}

        \filldraw[fill=blue!15, draw=white] (\ymin, 0) -- (\ymax, 0) -- (\ymin, \ydotmax) -- (\ymin, 0);
        \draw[blue, ultra thick, dotted] (\ymax, 0) -- (\ymin, \ydotmax);
        \node at (1.3*\ymin, 0.4*\ydotmax) {\textcolor{blue}{$\mathcal{B}$}};
        
        \node at (\width, \height) {$\mathcal{S}$};
        \draw[very thick, ->] (-\tick, 0) -- (\width, 0);
        \node at (\width+0.15, 0) {$y$};
        \draw[very thick, ->] (0, -0.45*\ydotmax) -- (0, \height);
        \node at (-0.25, \height) {$\dot y$};
        \node at (-0.25, 0) {$0$};
        \draw[ultra thick, dashed, red!60] (\ymax, -0.45*\ydotmax) -- (\ymax, -0.3);
        \draw[ultra thick, dashed, red!60] (\ymax, -0.15) -- (\ymax, \height);

        \draw[very thick, black] (-\tick, \ydotmax) -- (\tick, \ydotmax);
        \node at (-0.45, \ydotmax) {$\dot y_{max}$};
        \draw[very thick, black] (\ymin, -\tick) -- (\ymin, \tick);
        \node at (\ymin, -0.21) {$y_{min}$};
        \draw[very thick, black] (\ymax, -\tick) -- (\ymax, \tick);
        \node at (\ymax, -0.21) {\textcolor{red!60}{$y_{max}$}};

        \foreach \y in {0.3*\ydotmax, 0.5*\ydotmax, 0.7*\ydotmax, 0.9*\ydotmax, 1.1*\ydotmax}
            \def\len{0.7*\arrowLen * \y}
            \draw[thick, ->] (0.4*\ymin -0.5*\len, \y) -- (0.4*\ymin + 0.5*\len, \y);

        \foreach \y in {0.3*\ydotmax, 0.5*\ydotmax, 0.7*\ydotmax, 0.9*\ydotmax}
        {
            \def\len{0.7*\arrowLen * \y}
            \draw[thick, ->] (\ymin-0.5*\len, \y) -- (\ymin, \y) -- (\ymin + 0.5*\len, \y - 0.33*\len);
        }
        \foreach \y in {1.1*\ydotmax}
            \def\len{0.7*\arrowLen * \y}
            \draw[thick, ->] (\ymin-0.5*\len, \y) -- (\ymin + 0.5*\len, \y);
            
        \foreach \y in {0.2*\ydotmax, 0.4*\ydotmax, 0.6*\ydotmax}
            \def\len{0.7*\arrowLen * \y}
            \draw[thick, ->] (1.5*\ymin - 0.5*\len, \y) -- (1.5*\ymin + 0.5*\len, \y - 0.67*\len);
        \def\y{0.4*\ydotmax}
        \def\len{0.7*\arrowLen * \y}
        \draw[thick, ->] (1.8*\ymin - 0.5*\len, \y-0.18*\ydotmax) -- (1.8*\ymin + 0.5*\len, \y-0.18*\ydotmax - 0.67*\len);
        \def\y{0.6*\ydotmax}
        \def\len{0.7*\arrowLen * \y}
        \draw[thick, ->] (2*\ymin - 0.5*\len, \y-0.3*\ydotmax) -- (2*\ymin + 0.5*\len, \y-0.3*\ydotmax - 0.67*\len);

        \foreach \y in {0.9*\ydotmax, 1.1*\ydotmax}
            \def\len{0.7*\arrowLen * \y}
            \draw[thick, ->] (1.6*\ymin-0.5*\len, \y) -- (1.6*\ymin + 0.5*\len, \y);

        \foreach \y in {0.3*\ydotmax, 0.5*\ydotmax, 0.7*\ydotmax, 0.9*\ydotmax, 1.1*\ydotmax}
            \def\len{0.7*\arrowLen * \y}
            \draw[thick, ->] (2.2*\ymin-0.5*\len, \y) -- (2.2*\ymin + 0.5*\len, \y);
            
        \foreach \x in {1.3*\ymin, 1.62*\ymin, 1.94*\ymin, 2.25*\ymin}
            \draw[thick, ->] (\x, 0.08*\ydotmax) .. controls (\x+0.08*\ymin, 0) .. (\x, -0.08*\ydotmax);

        \foreach \y in {-0.2*\ydotmax, -0.4*\ydotmax}
            \def\len{0.1*\arrowLen + 0.7*\arrowLen * \y}
            \foreach \x in {0.4*\ymin, \ymin, 1.6*\ymin, 2.2*\ymin}
                \draw[thick, ->] (\x, \y) -- (\x - \len, \y);
    
    \end{tikzpicture}
    \caption{Phase portrait of constrained output $y$ illustrating our High Relative Degree POLICEd RL method on a system of relative degree $2$. To prevent states from violating constraint $y \leq y_{max}$ (\textcolor{red}{red dashed line}), our policy guarantees that trajectories entering buffer region $\mathcal{B}$ (\textcolor{blue!60}{blue}) cannot leave it through its upper bound (\textcolor{blue}{blue dotted line}). Our policy makes $\ddot y$ sufficiently negative in buffer $\mathcal{B}$ to bring $\dot y$ to $0$ in all trajectories entering $\mathcal{B}$. Once $\dot y < 0$, trajectories cannot approach the constraint.
    Due to the states' inertia, it is physically impossible to prevent all constraint violations. For instance, $y = y_{max}$, $\dot y >> 1$ will yield $y > y_{max}$ at the next timestep. Hence, we only aim at guaranteeing the safety of trajectories entering buffer $\mathcal{B}$.}
    \label{fig: schema}
\end{figure}

To learn our safe controller, we draw inspiration from the POLICEd RL method of \cite{POLICEd_RL} and transform the state space surrounding the affine state constraint into a buffer region that cannot be crossed. We overcome the major limitation of \cite{POLICEd_RL} by extending POLICEd RL to constraints of high relative degree.
To dissipate the inertia of this high relative degree constraint, our key insight is to extend the buffer into the dimensions of the unactuated derivatives of the state constraint.
While this task appears arduous, these derivatives are usually accessible since related to the states.
In this buffer region, we train the policy to dissipate the state's inertia to progressively slow its progression towards the constraint, as illustrated in Fig.~\ref{fig: schema}. This controller guarantees that trajectories entering the buffer do not violate the constraint. Since inertia cannot be dissipated instantly, some constraint violations are physically impossible to prevent.
Inspired by \cite{POLICEd_RL}, to easily verify the dissipative character of the controller in the buffer, we use the POLICE algorithm~\cite{police} to generate an affine policy over the buffer region.

In summary, our contributions in this work are as follows.

\begin{enumerate}[topsep=0pt, partopsep=0pt, leftmargin=10pt]
    \item We introduce High Relative Degree POLICEd RL, a novel RL framework to guarantee satisfaction of an affine state constraint of high relative degree using a black-box model of the system in closed-loop with a trained policy.
    \item We provide comprehensive proof, and we detail our method to evaluate our trained policy while directly guaranteeing constraint satisfaction.
    \item We demonstrate the safety guarantees of our approach in a number of simulation studies involving an inverted pendulum and a space shuttle landing.
\end{enumerate}

The remainder of this work is organized as follows.
In Section~\ref{sec: framework}, we introduce our problem formulation along with our framework.
In Section~\ref{sec: constraint}, we establish the theoretical guarantees of our approach in enforcing the satisfaction of a high relative degree constraint.
In Section~\ref{sec: simulations}, we illustrate our method on the Gym inverted pendulum and on a space shuttle landing scenario. Appendices~\ref{subsec: SSID} and \ref{subsec: lemmata} contain supporting lemmata and implementation details for our simulations.

\textit{Notation:}
We denote the positive integer interval from $a \in \mathbb{N}$ to $b \in \mathbb{N}$ inclusive by $[\![a, b]\!]$.
We denote the component $i \in [\![1,n]\!]$ of a vector $x \in \mathbb{R}^n$ by $x_i$ and the vector composed of components $i$ to $j > i$ by $x_{i:j}$. 
The set of nonnegative real numbers is $\mathbb{R}^+$.
We denote the $k^{th}$ time derivative of a function $y$ by $y^{(k)} = \frac{d^k y}{dt^k}$.
If $x, y \in \mathbb{R}^n$, then $x \leq y$ denotes the element-wise inequalities $x_i \leq y_i$ for all $i \in [\![1,n]\!]$.

\section{Framework}\label{sec: framework}

We consider a black-box deterministic system
\begin{equation}\label{eq: nonlinear dynamics}
    \dot x(t) = f\big( x(t), u(t) \big), \quad u(t) \in \mathcal{U}, \quad x(0) \in \mathcal{X},
\end{equation}
with state space $\mathcal{X} \subseteq \mathbb{R}^n$ and admissible control set $\mathcal{U} \subseteq \mathbb{R}^m$. 
We consider dynamics \eqref{eq: nonlinear dynamics} to be an implicit black-box, meaning that we can evaluate $f$ but we do not have any explicit knowledge or analytical form of $f$. This is similar to the online RL setting where $f$ is a simulator or a robot.

We assume that the system safety constraint is captured by a single affine inequality on output
\begin{equation}\label{eq: constraint}
    y(t) := C x(t) \leq y_{max}, \quad \text{for all}\ t \geq 0,
\end{equation}
with $C \in \mathbb{R}^{1 \times n}$ and $y_{max} \in \mathbb{R}$.
We model our deterministic feedback policy $u(t) = \pi_\theta\big( x(t) \big) \in \mathcal{U}$ by a deep neural network parameterized by $\theta$. Our objective is to train policy $\pi_\theta$ to respect constraint~\eqref{eq: constraint} and maximize the following expected discounted reward 
\begin{equation}\label{eq: expected reward}
    \underset{\theta}{\max}\, \mathcal{G}(\pi_\theta) := \hspace{-2mm} \underset{x_0 \sim \rho_0}{\mathbb{E}} \hspace{-1mm} \int_0^{\infty} \hspace{-3mm} \gamma^t R\big( x(t), \pi_\theta(x(t))\big) dt \hspace{2mm} \text{s.t.}\ \eqref{eq: constraint}, \hspace{2mm}
\end{equation}
where $\gamma \in (0,1]$ is a discount factor, $R$ a reward function, and $\rho_0$ the distribution of initial states. The only stochasticity in our setting comes from the initial state sampling $x_0 \sim \rho_0$. We emphasize that constraint~\eqref{eq: constraint} is a \textit{hard constraint} to be respected at all times.
Contrary to previous work~\cite{POLICEd_RL}, we assume that constraint~\eqref{eq: constraint} has a \emph{relative degree} at least $2$.

\begin{definition}\label{def: relative degree}
    The \emph{relative degree} $r$ of output $y$ \eqref{eq: constraint} for dynamics~\eqref{eq: nonlinear dynamics} is the order of its input-output relationship, i.e., $r := \min\big\{ p \in \mathbb{N} : \frac{\partial}{\partial u} \frac{\partial^p y}{\partial t^p}(t) \neq 0 \ \text{for all}\ x \in \mathcal{X} \big\}$ \emph{\cite{book_nonlinear}}.
\end{definition}

In simpler words, the relative degree is the minimal number of times output $y$ has to be differentiated until control input $u$ appears. As argued in \cite{feedback_linearization}, relative degree $r$ can be obtained by first-order principles without the knowledge of dynamics $f$. Hence, knowing $r$ is compatible with our black-box model of $f$.
Assuming $r \geq 2$, we have
\begin{equation}\label{eq: rel deg 1}
    \frac{\partial}{\partial u} \frac{\partial y(t)}{\partial t} = \frac{\partial \dot y(t)}{\partial u} = \frac{\partial}{\partial u} C \dot x(t) = C \frac{\partial f(x,u)}{\partial u} = 0,
\end{equation}
for all $x \in \mathcal{X}$. Taking one further time derivative yields
\begin{align*}
    \ddot y(t) &= C \ddot x(t) = C \frac{\partial f(x,u)}{\partial t} \\
    &= C \frac{\partial f(x,u)}{\partial x}\frac{\partial x}{\partial t}  + \underbrace{C \frac{\partial f(x,u)}{\partial u}}_{=\, 0\ \text{from}\ \eqref{eq: rel deg 1}} \frac{\partial u}{\partial t} = C \mathcal{D}\! f(x,u),
\end{align*}
where $\mathcal{D}\! f(x,u) := \frac{\partial f(x,u)}{\partial x} f(x,u)$ differs from a Lie derivative since $f$ depends not only on $x$ but also on $u$ since \eqref{eq: nonlinear dynamics} is \emph{not control affine}. 
Iterating this process yields 
\begin{equation*}
    y^{(k)}(t) = C \mathcal{D}^{k-1}\! f(x,u)
\end{equation*}
for all $k \in [\![0, r-1]\!]$ with 
\begin{equation*}
    \mathcal{D}^{k+1}\! f := \frac{\partial \mathcal{D}^k\!f}{\partial x} f \quad \text{and} \quad \mathcal{D}^0\! f := f.
\end{equation*}

Having $r \geq 2$ means that $u$ does not appear in the expression of $\dot y$. Thus, a change in control input will not immediately modify $y$. We follow \cite{ZCBF} and refer to the unactuated derivatives of $y$ $\big( y^{(k)}$ for $k \in [\![1, r-1]\!] \big)$ as generalized inertia, by analogy to inertia in kinematic systems. To enforce constraint \eqref{eq: constraint}, we need to dissipate this generalized inertia before reaching constraint line $y = y_{max}$. To easily assess this generalized inertia, we make the following assumption.

\begin{assumption}\label{assum: T}
    There exists an invertible map $T$ between state $x \in \mathbb{R}^n$ and $s \in \mathbb{R}^n$, whose first $r$ components are $s_1 = y$, $s_2 = \dot y$, ..., $s_r = y^{(r-1)}$, where $y$ is output~\eqref{eq: constraint}.
    Transformation $T(x) = s$ gives rise to an equivalent state space $\mathcal{S} := T(\mathcal{X})$.
\end{assumption}

Note that this is a rather mild assumption. Indeed, a transformation $s = T(x)$ always exists since $y = Cx$ and thus $s_{k+2} = y^{(k+1)} = C \frac{\partial^k \dot x}{\partial t} = C \frac{\partial^k f}{\partial t}(x,u)$. Assumption~\ref{assum: T} is required for the invertibility of $T$ and the fact that $T$ can be determined without knowledge of black-box dynamics $f$. 
For typical control systems as studied in Section~\ref{sec: simulations}, $T$ is a simple function satisfying Assumption~\ref{assum: T}.
Following Assumption~\ref{assum: T}, we now have two equivalent state representations: $x$ denotes the original state of system \eqref{eq: nonlinear dynamics}, while $s$ denotes the transformed state composed of the iterated derivatives of constrained output $y$.

We can now formally define our problem of interest.

\begin{problem}\label{prob: continuous}
    Given:
    \begin{enumerate}[topsep=0pt, parsep=1pt, partopsep=-2pt]
        \item black-box control system \eqref{eq: nonlinear dynamics};
        \item state space $\mathcal{X} \subseteq \mathbb{R}^n$;
        \item admissible input set $\mathcal{U} \subseteq \mathbb{R}^m$;
        \item affine constraint \eqref{eq: constraint} of relative degree $r \geq 2$;
        \item neural network policy $\pi_\theta(x)$ parameterized by $\theta$;
        \item invertible transformation $T$ of Assumption~\ref{assum: T};
    \end{enumerate}
    Our goal is to solve $\theta^* = \underset{\theta}{\arg\max}\ \mathcal{G}(\pi_\theta)$ s.t. \eqref{eq: nonlinear dynamics} and \eqref{eq: constraint}.
\end{problem}

\section{Constrained Reinforcement Learning}\label{sec: constraint}

In this section, we devise a method to solve Problem~\ref{prob: continuous} by designing a buffer preventing trajectories from breaching the constraint, as illustrated in Fig.~\ref{fig: schema}.
To enforce high relative degree constraint $y \leq y_{max}$, our safe controller must dissipate its generalized inertia before reaching constraint line $y = y_{max}$. We force this dissipation in a buffer region $\mathcal{B}$ as illustrated in Fig.~\ref{fig: schema}. We will first build such a buffer, then show that trajectories entering $\mathcal{B}$ cannot breach the constraint.

\subsection{Buffer Design}\label{subsec: buffer}

We formalize the concept of Fig.~\ref{fig: schema} and build buffer $\mathcal{B}$ as the state space region where our controller will dissipate generalized inertia $\dot y$,...,$y^{(r-1)}$ to prevent violation of constraint~\eqref{eq: constraint}. To adapt $\mathcal{B}$ to this task, we design it in state space $\mathcal{S}$, because Assumption~\ref{assum: T} states that the coordinates of $\mathcal{S}$ are the generalized inertia components $y^{(k)}$. 

Assume we can choose $s \in \mathcal{B}$ and let us investigate what bounds the components of $s$ should satisfy to remain in $\mathcal{B}$. Following Assumption~\ref{assum: T}, the first component of $s$ is $s_1 = y$, and thus should satisfy $s_1 \leq y_{max}$ to respect constraint~\eqref{eq: constraint}. We choose a lower bound for $s_1$ as $y_{min} < y_{max}$. 

Following Assumption~\ref{assum: T}, the second component of $s$ is $s_2 = \dot y$. To maintain $y \leq y_{max}$, we need $\dot y \leq 0$ when $y = y_{max}$.
Requiring $\dot y \leq 0$ for all $s \in \mathcal{B}$ is the approach of \cite{POLICEd_RL} but restricts $\mathcal{B}$ to only include states already moving away from upper bound $y_{max}$, whereas we want to slow down and stop trajectories going towards $y_{max}$. Thus, we must allow states with $\dot y > 0$ in $\mathcal{B}$. Let $\dot y_{max} > 0$ be the maximal velocity in $\mathcal{B}$. Our controller will later require $\mathcal{B}$ to be a polytope. Thus, we naturally define the upper bound on $\dot y = s_2$ as 
\begin{equation}\label{eq: beta}
    s_2^{max} := \beta (y_{max} - y)  \quad \text{with} \quad \beta := \frac{\dot y_{max}}{y_{max} - y_{min}},
\end{equation}
so that $s_2^{max}(y) = \dot y_{max}$ when $y = y_{min}$ and $s_2^{max} = 0$ when $y = y_{max}$, as illustrated in Fig.~\ref{fig: schema}. We choose a lower bound $s_2^{min} \leq 0$ so that $s_2^{min} \leq s_2^{max}(y)$ for all $y$.
Note that 
\begin{equation}\label{eq: y dif ineq}
    s_2 = \dot y \leq s_2^{max}(y) = \beta(y_{max} - y)   
\end{equation}
is a differential inequality on $y$ that we designed to maintain $y \leq y_{max}$. To enforce \eqref{eq: y dif ineq} we need a control input, but only $y^{(r)}$ is actuated. Our key idea is then to make our controller enforce a differential inequality on $y^{(r)}$, whose iterated integrations will lead to \eqref{eq: y dif ineq}.

Working backwards, we differentiate \eqref{eq: y dif ineq} into $\ddot y \leq -\beta \dot y$, i.e., $s_3 \leq -\beta s_2$. Thus, we choose bounds $s_3^{min} < s_3^{max}(s) := -\beta s_2$. Iterating this process until $y^{(r)}$ leads to lower and upper bounds $\underline{b}$ and $\overline{b}$ on the first $r$ components of $s \in \mathcal{B}$. Then, $s_{1:r} \in \big[ \underline{b},\, \overline{b}(s) \big]$ element-wise with
\begin{align}
    \underline{b} &:= \big[ y_{min},\hspace{8.5mm} s_2^{min},\hspace{8.5mm} s_3^{min},\ \hdots,\hspace{3.5mm} s_{r}^{min} \hspace{2mm}\big], \label{eq: buffer lower bound}\\
    \overline{b}(s) &:= \big[ y_{max},\ \beta (y_{max} -s_1),\ -\beta s_2,\ \hdots,\ -\beta s_{r-1}\big]. \label{eq: buffer upper bound}
\end{align}
The remaining $n-r$ components of $s \in \mathcal{B}$ are not derivatives of $y$ and thus are not involved in the constraint enforcement process. As mentioned previously, we will need $\mathcal{B}$ to be a polytope, hence we choose to bound the last $n-r$ components of $s \in \mathcal{B}$ by a polytope $\mathcal{P} \subset \mathbb{R}^{n-r}$ so that
\begin{equation}\label{eq: buffer}
    \mathcal{B} := \left\{ s \in \mathcal{S}: s_{1:r} \in \big[ \underline{b},\ \overline{b}(s) \big],\ s_{r+1:n} \in \mathcal{P} \right\}.
\end{equation}
Note that the bounds we just derived only delimit region $\mathcal{B}$ in $\mathcal{S}$, but without adequate control input, trajectories will not respect these bounds.
Similarly, the differential inequalities obtained above only reflect the desired dynamics that we want to enforce with our controller.
On the other hand, bounds $s_{1:r} \geq \underline{b}$ and $s_{r+1:n} \in \mathcal{P}$ will not be specifically enforced by the controller, but should be designed to encompass all trajectories to be safeguarded from $y_{max}$.

By design, buffer $\mathcal{B}$ of \eqref{eq: buffer} is then a compact convex polytope with a finite number $N$ of vertices gathered in the set $\mathcal{V}\big(\mathcal{B}\big) := \big\{v^1, \hdots, v^{N} \big\}$\footnote{The number of vertices of $\mathcal{B}$ is related to the Fibonacci sequence (see Lemma~\ref{lemma: number vertices}).}.

\subsection{Controller Design}

Let us now design a controller to maintain trajectories in buffer $\mathcal{B}$. Inspired by \cite{POLICEd_RL}, we model our control policy with a POLICEd neural network $\mu_\theta := \pi_\theta \circ T : \mathcal{S} \rightarrow \mathcal{U}$, with continuous piecewise affine activation functions such as ReLU \cite{police}.
This restriction is of little concern as ReLU is the most commonly used activation function.
This POLICEd architecture allows us to make the outputs of $\mu_\theta$ affine over a polytopic region of state space $\mathcal{S}$, which we chose to be $\mathcal{B}$. Then, there exist matrices $D_{\theta} \in \mathbb{R}^{m \times n}$ and $e_{\theta} \in \mathbb{R}^m$ such that
\begin{equation}\label{eq: police}
    \mu_\theta(s) = D_{\theta} s + e_{\theta} \quad \text{for all}\ s \in \mathcal{B}.
\end{equation}

Since buffer $\mathcal{B}$ and policy $\mu_\theta$ are both in space $\mathcal{S}$ and not $\mathcal{X}$, we need to calculate the state dynamics in $\mathcal{S}$. State $s = T(x)$ following controller $\mu_\theta$ satisfies
\begin{equation*}
    \dot s = \frac{\partial T}{\partial t}(x) = \frac{\partial T}{\partial x}\frac{\partial x}{\partial t} = \frac{\partial T}{\partial x}(x) f\Big(x,\, \mu_\theta\big( T(x) \big) \Big).
\end{equation*}
Then, by defining the map
\begin{equation*}
    \widetilde{f}(s; \mu_\theta) := \frac{\partial T}{\partial x}\big( T^{-1}(s)\big) f\big( T^{-1}(s),\, \mu_\theta(s)\big),
\end{equation*}
we can write $\dot s(t) = \widetilde{f}\big( s(t); \mu_\theta \big)$.
We are mostly interested in the dynamics of the actuated derivative of output $y$:
\begin{equation}\label{eq: y^(r)}
    y^{(r)}(t) = \frac{\partial y^{(r-1)}}{\partial t}(t) = \frac{\partial s_r}{\partial t}(t) = \widetilde{f}_r\big( s(t); \mu_\theta\big),
\end{equation}
where $s_r$ and $\widetilde{f}_r$ denote the $r^{th}$ component of $s$ and $\widetilde{f}$ respectively.
If dynamics~\eqref{eq: y^(r)} were known and affine, their coupling with affine policy $\mu_\theta$ on $\mathcal{B}$ would lead to a simple constraint enforcement process.
However, dynamics~\eqref{eq: y^(r)} are a black-box and possibly nonlinear. We will thus use an affine approximation of \eqref{eq: y^(r)} inside using the following definition.

\begin{definition}\label{def: epsilon}
    An \emph{approximation measure} $\varepsilon$ of dynamics~\eqref{eq: y^(r)} in buffer~\eqref{eq: buffer} is any $\varepsilon \geq 0$ for which there exists any matrices $A \in \mathbb{R}^{n \times n}$, $B \in \mathbb{R}^{n \times m}$, and $c \in \mathbb{R}^n$ such that
    \begin{equation}\label{eq: approximation}
        \big| \widetilde{f}_r\big( s; \mu_\theta\big) - C\big(A s + B \mu_\theta(s) + c\big) \big| \leq \varepsilon,
    \end{equation}
    for all $s \in \mathcal{B}$.
\end{definition}

Intuitively, the value of $\varepsilon$ quantifies how far from affine is function $\widetilde{f}_r$ over buffer $\mathcal{B}$. 
Having access to map $T$, controller $\mu_\theta$ and to a black-box model of $f$, we can evaluate $\widetilde{f}_r$ and compute $\varepsilon$ using linear least square approximation \cite{linear_least_squares}. Since $\varepsilon$ is estimated from data, it might not verify~\eqref{eq: approximation} for some $s \in \mathcal{B}$ absent from the dataset. To satisfy Definition~\ref{def: epsilon}, we need to \emph{over}-approximate $\varepsilon$ since any upper bound will verify \eqref{eq: approximation}. With such an upper bound, we will guarantee the satisfaction of constraint~\eqref{eq: constraint} with actual dynamics~\eqref{eq: nonlinear dynamics}.

We now establish our central result demonstrating how to \emph{guarantee} satisfaction of constraint \eqref{eq: constraint} by black-box environment~\eqref{eq: nonlinear dynamics} armed \emph{only} with an approximation measure $\varepsilon$ and \emph{without} knowing $A$, $B$, $c$, $f$, or $\widetilde{f}_r$.

\begin{theorem}\label{thm: admissible trajectories}
    Assume that for some approximation measure $\varepsilon$, dissipation condition 
    \begin{equation}\label{eq: dissipation}
        \widetilde{f}_r(v; \mu_\theta) \leq -2\varepsilon -\beta v_r, 
    \end{equation}
    holds for all $v \in \mathcal{V}\big(\mathcal{B}\big)$, where $v_r$ is the $r^{th}$ component of $v$ and $\beta$ comes from \eqref{eq: beta}. If a trajectory $s$ steered by $\mu_\theta$ verifies 
    \begin{equation}\label{eq: strictly in B}
        s_{1:r}(t_0) < \overline{b}\big( s(t_0) \big)   
    \end{equation}
    for some $t_0 \geq 0$, and satisfies 
    \begin{equation}\label{eq: stays in B}
        s_{1:r}(t) \geq \underline{b} \quad \text{and} \quad s_{r+1:n}(t) \in \mathcal{P}
    \end{equation}
    for all $t \in [t_0, t_1)$, then $s_{1:r}(t) < \overline{b}\big( s(t) \big)$ for all $t \in [t_0, t_1)$.
\end{theorem}

In simpler words, Theorem~\ref{thm: admissible trajectories} guarantees that trajectories entering buffer $\mathcal{B}$ below upper bound $\overline{b}$ cannot exit $\mathcal{B}$ through $\overline{b}$ as long as dissipation condition~\eqref{eq: dissipation} is satisfied. Theorem~\ref{thm: admissible trajectories} generates the bent arrows of the flow illustrated in Fig.~\ref{fig: schema} which prevent trajectories from violating constraint~\eqref{eq: constraint}. The major strength of our approach is that dissipation condition \eqref{eq: dissipation} only needs to be enforced at the vertices of $\mathcal{B}$ and thus does not require knowledge of $f$ or $\widetilde{f}_r$.

\begin{proof}[Proof of Theorem~\ref{thm: admissible trajectories}]
    The intuition behind this proof is to use the convexity of buffer $\mathcal{B}$ and affine approximation~\eqref{eq: approximation} to extend condition~\eqref{eq: dissipation} to the entire $\mathcal{B}$. By combining this condition with the specific design of upper bound~\eqref{eq: buffer upper bound}, we can derive bounds on the output derivatives $y$, ..., $y^{(r-1)}$ and show that they cannot cross upper bound $\overline{b}$.

    We divide the proof into three lemmas. First, we show in Lemma~\ref{lemma: convex} that condition \eqref{eq: dissipation} yields $\dot s_r \leq -\beta s_r$ for all $s$ in $\mathcal{B}$.
    This condition combined in Lemma~\ref{lemma: B closed} with \eqref{eq: stays in B} yields the differential inequalities to be respected by $y$, $\dot y$,..., $y^{(r-1)}$ as long as trajectory $s$ remains in $\mathcal{B}$.
    In Lemma~\ref{lemma: B open} these differential equations are then paired with initial conditions \eqref{eq: strictly in B} to obtain $s_{1:r}(t) < \overline{b}\big( s(t) \big)$ for all $t \in [t_0, t_1)$.
\end{proof}

\subsection{Supporting Lemmata}

We now extend dissipation condition~\eqref{eq: dissipation} from the vertices of buffer $\mathcal{B}$ to the whole set $\mathcal{B}$.

\begin{lemma}\label{lemma: convex}
    If for some approximation measure $\varepsilon$, condition~\eqref{eq: dissipation} holds for all $v \in \mathcal{V}(\mathcal{B})$, then controller $\mu_\theta$ yields 
    \begin{equation}\label{eq: y r repulsion}
        \dot s_r(t) \leq -\beta s_r(t), \quad \text{i.e.,} \quad y^{(r)}(t) \leq -\beta y^{(r-1)}(t),
    \end{equation}
    for all $s(t) \in \mathcal{B}$.
\end{lemma}
\begin{proof}
    Since $\varepsilon$ is an approximation measure, there exist $A$, $B$ and $c$ verifying \eqref{eq: approximation} which we evaluate at $s = v \in \mathcal{V}\big(\mathcal{B}\big)$
    \begin{align}\label{eq: vertex repulsion}
        &C \big( Av + B\mu_\theta(v) + c \big) \nonumber \\
        &\quad \leq \big|C\big( Av + B\mu_\theta(v) + c\big) - \widetilde{f}_r( v; \mu_\theta)\big| + \widetilde{f}_r( v; \mu_\theta) \nonumber \\
        &\quad \leq \varepsilon -2\varepsilon -\beta v_r,
    \end{align}
    where the first inequality is a triangular inequality, the second follows from \eqref{eq: approximation} and \eqref{eq: dissipation}.
    Using the convexity of polytope $\mathcal{B}$ of vertices $\mathcal{V}\big(\mathcal{B}\big) = \big\{v^1, \hdots, v^{N}\big\}$, for any $s \in \mathcal{B}$, there exist $\alpha^1, \hdots, \alpha^{N} \in \mathbb{R}^+$ such that $\sum_{k = 1}^{N} \alpha^k = 1$ and $s = \sum_{k=1}^{N} \alpha^k v^k$. Controller \eqref{eq: police} applied at $s \in \mathcal{B}$ yields
    \begin{align}\label{eq: vertices eps}
        C \big( As + &B\mu_\theta(s) + c \big) = C \big( As + B(D_{\theta} s + e_{\theta}) + c \big) \nonumber \\
        &= C(A + B D_{\theta})s + C(B e_{\theta} + c) \nonumber  \\
        &= C (A + B D_{\theta})\sum_{k=1}^{N} \alpha^k v^k + C (B e_{\theta} + c)\sum_{k = 1}^{N} \alpha^k \nonumber  \\
        &= \sum_{k = 1}^{N} \alpha^k C \big( (A + B D_{\theta})v^k + B e_{\theta} + c \big) \nonumber \\
        &= \sum_{k = 1}^{N} \alpha^k C \big( A v^k + B \mu_\theta(v^k) + c \big) \nonumber \\
        &\leq \sum_{k = 1}^{N} \alpha^k \big(-\varepsilon - \beta v_r^k \big) = -\varepsilon \sum_{k = 1}^{N} \alpha^k - \beta \sum_{k = 1}^{N} \alpha^k v_r^k \nonumber \\
        &= -\varepsilon -\beta s_r,
    \end{align}
    where the only inequality comes from \eqref{eq: vertex repulsion} on each vertex $v^k$ and the last equality stems from the linear decomposition of component $r$ of state $s$ between component $r$ of vertices $v^k$. For any state $s \in \mathcal{B}$, \eqref{eq: y^(r)} yields
    \begin{align*}
        y^{(r)} &= \widetilde{f}_r(s, \mu_\theta) \\
        &\hspace{-1.5mm}\leq \big|\widetilde{f}_r(s, \mu_\theta)\! -\! C\big(\! A s\! +\!\! B\mu_\theta(s)\! +\! c\big)\! \big|\! +\! C\big(\! A s\! +\!\! B \mu_\theta(s) + c \big) \\
        &\hspace{-1.5mm}\leq \varepsilon - \varepsilon - \beta s_r = -\beta y^{(r-1)},
    \end{align*}
    where we first use the triangular inequality, then \eqref{eq: approximation} and \eqref{eq: vertices eps}, and the last equality comes from the definition of state $s$ in Assumption~\ref{assum: T}.
\end{proof}

Lemma~\ref{lemma: convex} uses the convexity of $\mathcal{B}$ and affine approximation~\eqref{eq: approximation} to extend \eqref{eq: dissipation}, valid \emph{only} at the vertices of $\mathcal{B}$, into \eqref{eq: y r repulsion}, valid all over $\mathcal{B}$. Without the POLICE algorithm \cite{police}, $\mu_\theta$ would not be affine over $\mathcal{B}$, and dissipation condition~\eqref{eq: dissipation} would need to be enforced everywhere on the buffer at a prohibitive computational cost.

\begin{lemma}\label{lemma: B closed}
    If \eqref{eq: y r repulsion} holds for all $s \in \mathcal{B}$, $s(t_0) \in \mathcal{B}$, and \eqref{eq: stays in B} holds for all $t \in [t_0, t_1)$ for some $t_1 > t_0 \geq 0$, then   
    \begin{equation}\label{eq: y solution}
        y(t) \leq \big(y(t_0) - y_{max}\big)e^{-\beta (t - t_0)} + y_{max}
    \end{equation}
    and 
    \begin{equation}\label{eq: y^k solution}
        y^{(k)}(t) \leq y^{(k)}(t_0) e^{-\beta (t-t_0)}
    \end{equation}
    for all $k \in [\![1, r-1]\!]$ and all $t \in [t_0, t_1)$.
\end{lemma}   
\begin{proof}   
    We apply the comparison lemma of \cite{Khalil} to differential inequality \eqref{eq: y r repulsion}, which yields $y^{(r-1)}(t) \leq y^{(r-1)}(t_0)e^{-\beta (t-t_0)}$.
    
    Initial condition $s(t_0) \in \mathcal{B}$ yields 
    $$s_r(t_0) = y^{(r-1)}(t_0) \leq -\beta s_{r-1}(t_0) = -\beta y^{(r-2)}(t_0).$$
    Define function $g(t) := y^{(r-1)}(t) + \beta y^{(r-2)}(t)$. Then, \eqref{eq: y r repulsion} is equivalent to $\dot g(t) \leq 0$ and our initial condition is $g(t_0) \leq 0$. Therefore, $g(t) \leq 0$ for all $t \in [t_0, t_1)$, i.e., $y^{(r-1)}(t) \leq -\beta y^{(r-2)}(t)$. Using the comparison lemma of \cite{Khalil}, we can solve this differential inequality and obtain $y^{(r-2)}(t) \leq y^{(r-2)}(t_0)e^{-\beta (t-t_0)}$.

    We can iterate this process for $k \in \{r-3, \hdots, 1\}$ and obtain $y^{(k+1)}(t) \leq -\beta y^{(k)}(t)$, which yields $y^{(k)}(t) \leq y^{(k)}(t_0) e^{-\beta (t-t_0)}$ for all $t \in [t_0, t_1)$.
    
    For $k = 1$, we thus have $\ddot y(t) + \beta \dot y(t) \leq 0$. Initial condition $s(t_0) \in \mathcal{B}$ yields $\dot y(t_0) \leq \dot \beta \big(y_{max} - y(t_0)\big)$, or equivalently $\dot y(t_0) + \beta y(t_0) \leq \beta y_{max}$. As previously, with $g(t) := \dot y(t) + \beta y(t)$, we have $\dot g(t) \leq 0$ and $g(t_0) \leq \beta y_{max}$. Thus, $g(t) \leq \beta y_{max}$, i.e., $\dot y(t) + \beta y(t) \leq \beta y_{max}$ for all $t \in [t_0, t_1)$. We solve this differential inequality using the comparison lemma of \cite{Khalil} and obtain \eqref{eq: y solution} for all $t \in [t_0, t_1)$.
\end{proof}

Lemma~\ref{lemma: B closed} used \eqref{eq: y r repulsion} and upper bound $\overline{b}$ of \eqref{eq: buffer upper bound} to obtain the differential equations verified by the output derivatives in $\mathcal{B}$. We will now use initial condition \eqref{eq: strictly in B} to show that trajectories cannot leave $\mathcal{B}$ through its upper bound $\overline{b}$.

\begin{lemma}\label{lemma: B open}
    If \eqref{eq: y r repulsion} holds for all $s \in \mathcal{B}$, \eqref{eq: y solution} and \eqref{eq: y^k solution} hold for all $t \in [t_0, t_1)$, then \eqref{eq: strictly in B} implies $s_{1:r}(t) \leq \overline{b}(s(t))$ for all $t \in [t_0, t_1)$.
\end{lemma}
\begin{proof}
    Condition $s(t_0) < \overline{b}(s(t_0))$ yields $y(t_0) < y_{max}$. This initial condition combined with \eqref{eq: y solution} leads to $y(t) < y_{max}$ for all $t \in [t_0, t_1)$, i.e., $s_1(t) < \overline{b}(s(t))_1$.
    
    \noindent Initial condition $s_2(t_0) < \overline{b}(s(t_0))_2$ yields $\dot y(t_0) < \beta(y_{max} - y(t_0))$. Starting from \eqref{eq: y^k solution} for $k = 1$, we have
    \begin{equation}\label{eq: y dot}
        \dot y(t) \leq \dot y(t_0)e^{-\beta (t-t_0)} < \beta(y_{max} - y(t_0)) e^{-\beta (t-t_0)}
    \end{equation}
    for all $t \in [t_0, t_1)$. Reorganizing \eqref{eq: y solution} leads to $-y(t_0)e^{-\beta (t-t_0)} \leq -y(t) + y_{max}(1 - e^{-\beta (t-t_0)})$, which can be combined with \eqref{eq: y dot} into
    \begin{align*}
        \dot y(t) &< \beta \big( y_{max}e^{-\beta (t-t_0)}-y(t) + y_{max}(1 - e^{-\beta (t-t_0)}) \big) \\
        &< \beta \big( y_{max} - y(t)\big) e^{-\beta (t-t_0)},
    \end{align*}
    i.e., $s_2(t) < \overline{b}(s(t))_2$ for all $t \in [t_0, t_1)$.

    Similarly for $k \in [\![2, r-1]\!]$, \eqref{eq: y^k solution} combined with initial condition $y^{(k)}(t_0) < -\beta y^{(k-1)}(t_0)$ leads to
     \begin{equation}\label{eq: y^k}
        y^{(k)}\!(t) \leq y^{(k)}\!(t_0)e^{-\beta (t-t_0)} < -\beta y^{(k-1)}\!(t_0) e^{-\beta (t-t_0)}
    \end{equation}
    for all $t \in [t_0, t_1)$. Reversing \eqref{eq: y^k solution} at $k-1$ leads to $-y^{(k-1)}(t_0)e^{-\beta (t-t_0)} \leq -y^{(k-1)}(t)$, which can be combined with \eqref{eq: y^k} into $y^{(k)}(t) < -\beta y^{(k-1)}(t)$, i.e., $s_{k+1}(t) < \overline{b}(s(t))_{k+1}$ for all $t \in [t_0, t_1)$.
\end{proof}

\begin{remark}
    Control set $\mathcal{U}$ of \eqref{eq: nonlinear dynamics} might prevent the existence of an admissible stabilizing policy $\mu_\theta$. That is why we use RL to find policy $\mu_\theta$ and we verify its safety with Theorem~\ref{thm: admissible trajectories}.
\end{remark}

Now that we have established our central result, we will illustrate its implementation on two numerical simulations.

\section{Numerical Simulations}\label{sec: simulations}

\subsection{Gym Inverted Pendulum}

We consider the Inverted Pendulum Gym environment~\cite{Gym} with the MuJoCo dynamics engine~\cite{mujoco} as illustrated in Fig.~\ref{fig: inverted pendulum}. The environment state $x$ is composed, in that order, of the cart position $p$, the pole angle $\theta$, and their derivatives $\dot p$ and $\dot \theta$.
The objective within this environment is to maintain the pole close to the vertical, i.e., $|\theta| \leq 0.2$ rad. Let us focus on enforcing the upper constraint, $y := \theta \leq 0.2$ rad. This constraint has a relative degree $r = 2$ since the control input is the force exerted on the cart, which directly impacts $\ddot \theta$.

\begin{figure}[t!]
    \centering
    \includegraphics[scale=0.25]{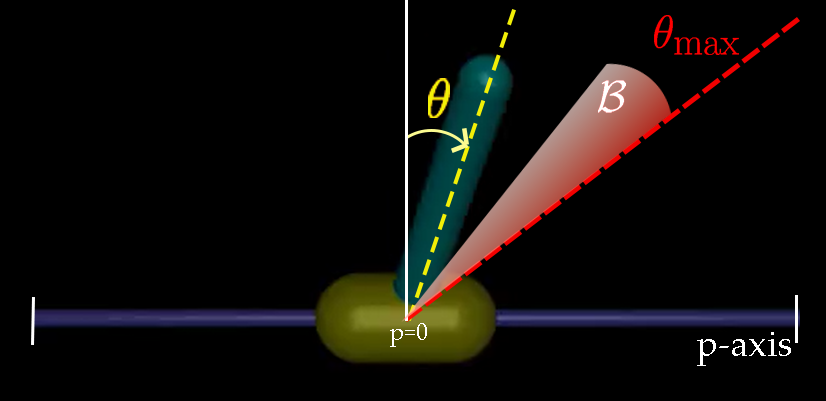}
    \caption{The inverted pendulum Gym environment \cite{Gym} annotated with cart position $p$, pendulum angle $\theta$, and buffer $\mathcal{B}$.}
    \label{fig: inverted pendulum}
\end{figure}

Following our choice of $y = \theta$, we define state $s = T(x)$ 
\begin{equation*}
    s := \begin{bmatrix} y \\ \dot y \\ s_3 \\ s_4 \end{bmatrix} = \begin{bmatrix} \theta \\ \dot \theta \\ p \\ \dot p \end{bmatrix} = \begin{bmatrix}
        0 & 1 & 0 & 0 \\
        0 & 0 & 0 & 1 \\
        1 & 0 & 0 & 0 \\
        0 & 0 & 1 & 0 \end{bmatrix} \begin{bmatrix} p \\ \theta \\ \dot p \\ \dot \theta \end{bmatrix} = Tx.
\end{equation*}
Transformation $T$ is linear and invertible, hence verifying Assumption~\ref{assum: T}. Additionally, we obtain $T$ without violating the black-box assumption on dynamics $f$, since $T$ only reorders the states.

Following Section~\ref{subsec: buffer}, we will now design buffer $\mathcal{B}$, whose architecture should help dissipate the inertia of trajectories arriving at $\theta = y_{min}$ with velocities $\dot \theta \leq \dot y_{max}$. We choose $y_{max} := 0.2$ rad, $y_{min} = 0.1$ rad, $s_2^{min} = 0$ rad/s and $\dot y_{max} = 1$ rad/s, and define buffer 
\begin{align*}
    \mathcal{B} := \big\{ s \in \mathcal{S} : s_1 &= y = \theta \in [0.1,\, 0.2],\\
    s_2 &= \dot y = \dot \theta \in [0,\, 2-10\theta],\\
    s_3 &= p \in [-0.9,\, 0.9],\\
    s_4 &= \dot p \in [-1,\, 1] \big\}.
\end{align*}
This choice of $\mathcal{B}$ allows only $\dot\theta = 0$ rad/s when $\theta = 0.2$ rad, hence preventing $\theta$ in $\mathcal{B}$ from growing past $0.2$ rad. Here polytope $\mathcal{P}$ of \eqref{eq: buffer} is $\mathcal{P} := [-0.9,\, 0.9] \times [-1,\, 1]$.

Following Definition~\ref{def: epsilon}, we sample states in $\mathcal{B}$ and perform a linear regression on $\ddot \theta$ to obtain an approximation measure $\varepsilon = 0.53$. We model controller $\mu_\theta$ with a deep neural network trained to stabilize the pole at $\theta = 0$ in a reinforcement learning fashion using proximal policy optimization (PPO) \cite{PPO}. We train two such policies, one being a standard multi-layer perceptron (MLP) to form a baseline, and the other having POLICEd layers \cite{police} enforcing affine condition~\eqref{eq: police}. Both policies follow the same training and are encouraged to enforce dissipation condition~\eqref{eq: dissipation} at the vertices of $\mathcal{B}$, which translates to $\ddot \theta \leq -2\varepsilon -10 \dot \theta$.

To illustrate Theorem~\ref{thm: admissible trajectories}, assume that our POLICEd controller $\mu_\theta$ enforces dissipation condition~\eqref{eq: dissipation} and let us consider a trajectory $s$ entering buffer $\mathcal{B}$ at time $t_0$. If initial state condition \eqref{eq: strictly in B} holds, i.e., if $\theta(t_0) < 0.2$ rad and $\dot \theta(t_0) < 2-10\theta(t_0)$, then as long as $\theta(t) \geq 0.1$ rad, $\dot \theta(t) \geq 0$ rad/s, and $\big( p(t), \dot p(t) \big) \in \mathcal{P}$, we have $\theta(t) < 0.2$ rad and $\dot \theta(t) < 2-10\theta(t)$. These equations generate the phase portrait of Fig.~\ref{fig: pendulum phase portrait}, which successfully reproduces the desired behavior exhibited in Fig.~\ref{fig: schema}. 

Our POLICEd controller guarantees that all trajectories entering $\mathcal{B}$ cannot cross its upper bound $\overline{b}$ and hence cannot violate the constraint, whereas some baseline trajectories cross $\overline{b}$ and violate the constraint as shown in Fig.~\ref{fig: pendulum phase portrait}.

\begin{figure}[t!]
    \centering
    \includegraphics[scale=0.5]{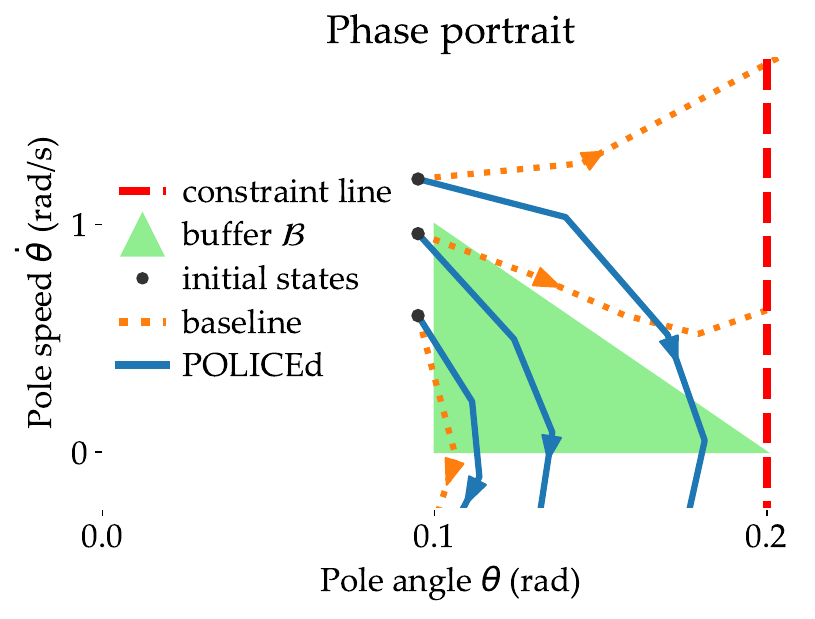}
    \caption{Phase portrait of $\big(\theta, \dot \theta \big)$ for the inverted pendulum. None of the POLICEd trajectories (\textcolor{blue!80}{blue}) entering buffer $\mathcal{B}$ (\textcolor{green!80!black}{green}) cross constraint line $\theta = 0.2$ rad (\textcolor{red}{dashed red}), whereas some of the baseline trajectories do (\textcolor{orange}{dotted orange}). Our approach guarantees that a pole arriving at $\theta = 0.1$ rad with a velocity $\dot \theta < 1$ rad/s will satisfy $\theta \leq 0.2$ rad. We do not guarantee the safety of POLICEd trajectories not entering the buffer.}
    \label{fig: pendulum phase portrait}
\end{figure}

\subsection{Space Shuttle Landing}\label{subsec: SSL}

We now study the highly-nonlinear dynamics of the space shuttle landing \cite{SpaceShuttle}. Original state $x \in \mathbb{R}^3$ is composed of the altitude $h$ of the shuttle, its flight path angle $\gamma$, and its velocity $v$, as seen in Fig.~\ref{fig: shuttle}. The dynamics of these state are
\begin{subequations}\label{eq: shuttle dynamics}
\begin{align}
    \dot h(t) &= v(t) \sin \gamma(t) \label{eq: h} \\
    \dot \gamma(t) &= \rho(t) v(t) C_L(t)\frac{S}{2m} - \frac{g \cos \gamma(t)}{v(t)} \label{eq: gamma} \\
    \dot v(t) &= -\rho(t) v^2(t) C_D(t)\frac{S}{2m} - g \sin \gamma(t), \label{eq: v}
\end{align}
\end{subequations}
where the air density satisfies $\rho(t) = \rho_0 e^{-h(t)/H}$ and the lift and drag coefficients take the form
\begin{equation}\label{eq: aero coefs}
\def\arraystretch{1.2}\begin{array}{l}
     C_L(t) = C_{L_0} \sin^2\! \alpha(t) \cos \alpha(t) \\
     C_D(t) = C_{D_0} + K C_L^2(t).
\end{array}
\end{equation}
The other parameters are detailed in Table~\ref{tab: shuttle values}. The control input is the angle-of-attack $\alpha$ of the shuttle, which makes these dynamics non-affine in control, and hence cannot be handled directly by any CBF method\footnote{Adding an integrator $\dot \alpha(t) = u(t)$ renders dynamics~\eqref{eq: shuttle dynamics} affine in control at the price of a higher relative degree \cite{Shuttle_thesis}.} \cite{ZCBF, Exponential_CBF, HOCBF}. 

\begin{figure}[t!]
    \centering
    \includegraphics[scale=0.12]{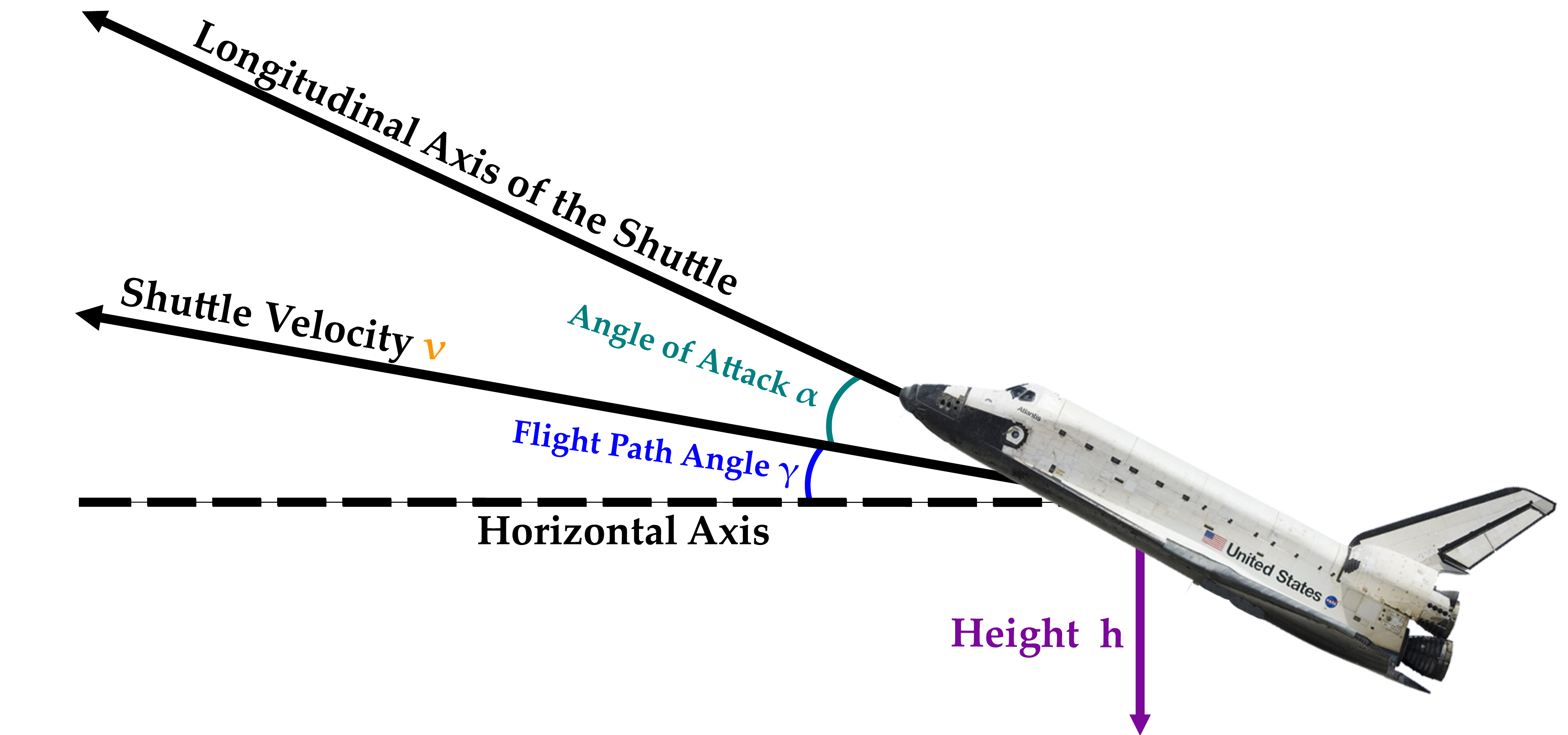}
    \caption{Illustration of our Space Shuttle environment. The state $x \in \mathbb{R}^3$ is composed of the altitude or height $h$ of the shuttle, its flight path angle $\gamma$, and its velocity $v$. The control action is the angle of attack $\alpha$.}
    \label{fig: shuttle}
\end{figure}

In this scenario, the shuttle starts from a descent configuration with high vertical velocity $\dot h$, which must be drastically reduced to allow for a soft landing.
More specifically, we consider initial states typical of a descent phase from a height $h_0 = 500$ ft, velocities $v_0 \in [300, 400]$ ft/s and flight path angles $\gamma_0 \in [-30^\circ, -10^\circ]$.
The objective of our controller is to bring the shuttle to a low altitude $h \leq 50$ ft with vertical velocity $\dot h \leq 6$ ft/s sufficiently small to allow for a soft landing \cite{SpaceShuttle}.
We choose an output constraint $y := -h \leq 0$, which has a relative degree $2$ for control input $\alpha$.
We build a buffer with $y_{min} = -50$ ft, $y_{max} = 0$ ft, $s_2^{min} = 6$ ft/s and $\dot y_{max} = 100$ ft/s.

We introduce state $s := \big(h, \dot h, \gamma \big)$ and invertible transformation $T(x) = s$ is easily obtainable from pure geometric considerations in Fig.~\ref{fig: shuttle} as it only needs \eqref{eq: h}. Thus, Assumption~\ref{assum: T} is verified and determining $T$ does not violate our black-box assumption on dynamics \eqref{eq: shuttle dynamics}, \eqref{eq: aero coefs}.

We train two PPO policies \cite{PPO} to minimize the vertical velocity at touchdown. One of these policies is a standard MLP used as a baseline and the other is our POLICEd version \cite{POLICEd_RL}. More implementation details are included in Appendix~\ref{subsec: SSID}.
As seen in Fig.~\ref{fig: shuttle phase} our POLICEd policy successfully enforces the dissipative buffer of Theorem~\ref{thm: admissible trajectories} and ensures soft landing of the shuttle contrary to the baseline PPO policy.

\begin{figure}[t!]
    \centering
    \includegraphics[scale=0.5]{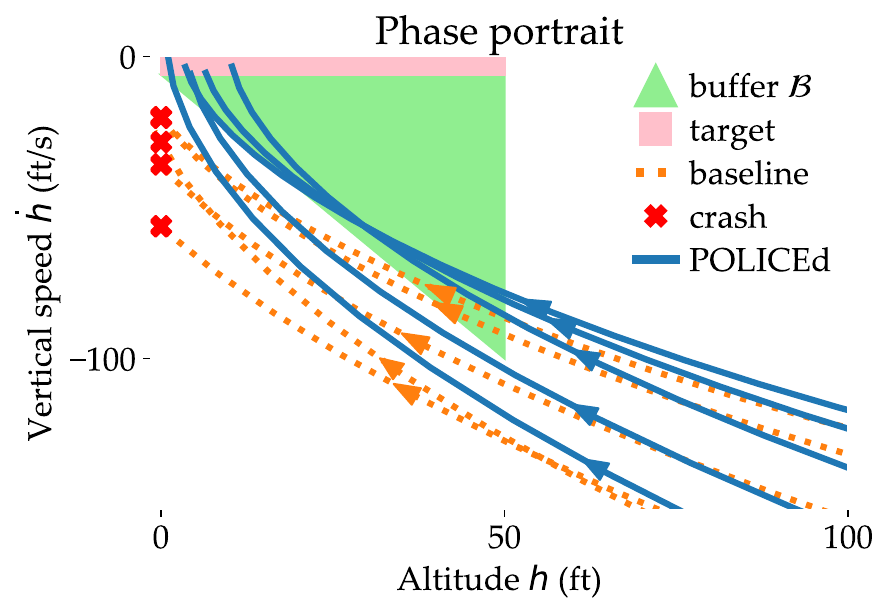}
    \caption{Phase portrait of the space shuttle landing. POLICEd trajectories (\textcolor{blue}{blue}) entering buffer $\mathcal{B}$ (\textcolor{green!80!black}{green}) all converge to a set of target conditions (\textcolor{pink!90!black}{pink}) with small vertical velocity from which landing is feasible. However, the baseline trajectories (\textcolor{orange}{dotted orange}) reach the ground $h = 0$ with high vertical velocities $\dot h \leq -6$ ft/s resulting in a crash of the shuttle (\textcolor{red}{\textbf{x}}).}
    \label{fig: shuttle phase}
\end{figure}

\section{Conclusion and Future Work}\label{sec: conclusion}

In this work, we established High Relative Degree POLICEd RL, a novel method to enforce a hard constraint of high relative degree on learned policies, while only using an implicit black-box model of the environment. We built a buffer region where the policy dissipates the generalized inertia of the high relative degree constraint to prevent trajectories from reaching the constraint line.
We illustrated our theory on the MuJoCo inverted pendulum and on a space shuttle landing scenario.

Several avenues for future work seem especially interesting. Extending the POLICE algorithm of \cite{police} to enforce multiple affine regions would allow a straightforward extension of this work to guarantee the satisfaction of multiple constraints of high relative degree.
Another interesting pursuit would be to investigate how to guarantee constraint satisfaction during the training process of the policy.

\appendix


\subsection{Space Shuttle Implementation Details}\label{subsec: SSID}

The baseline and POLICEd policies are both modeled by deep neural networks composed of 3 layers of 128 hidden units. Their reward function penalizes changes in the control input to encourage smooth variations of the angle of attack, and penalizes the final altitude and vertical velocity of the shuttle to promote soft landings. Its expression is
\begin{equation*}
    R(t) = -0.2 \big| a(t) - a(t-dt)\big| - \mathds{1}_{t = t_f} \big( |h(t_f)| + |\dot h(t_f)| \big),
\end{equation*}
where $\mathds{1}_{t = t_f}$ is the final time $t_f$ indicator function.

\begin{remark}\label{rmk: alpha input}
    Choosing $\alpha$ as input might seem unrealistic since $\alpha$ must be continuous. However, our control signal $u(t)$ is continuous by construction as a continuous function of the state $\mu_\theta(x(t))$.
    Adding an integrator $\dot \alpha(t) = u(t)$ as in \cite{Shuttle_affine} could make a reasonable input choice but at the price of increased complexity in calculating $T^{-1}$ to recover $\alpha$ from $s$.
\end{remark}

\begin{table}[t!]
    \centering
    \caption{Numerical values for the shuttle simulation from \cite{SpaceShuttle}.}
    \label{tab: shuttle values}
    \def\arraystretch{1.3}
    \begin{tabular}{ccc}
        \hline Parameter & Name & Value \\ \hline \hline
        $S/m$ & surface area over mass & $0.9118$ ft$^2$/slug \\
        $C_{L_0}$ & zero-angle-of-attack lift coefficient & $2.3$ \\
        $C_{D_0}$ & zero-lift drag coefficient & $0.0975$ \\
        $K$ & lift-induced drag coef. parameter & $0.1819$ \\
        $\rho_0$ & sea-level air density & $0.0027$ slugs/ft$^3$ \\
        $g$ & Earth’s gravitational acceleration & $32,174$ ft/s$^2$ \\
        $H$ & scale height & $27890$ ft \\ \hline
    \end{tabular}
\end{table}


\begin{figure*}[t!]
    \centering
    \begin{subfigure}[t]{0.32\textwidth}
        \includegraphics[scale=0.42]{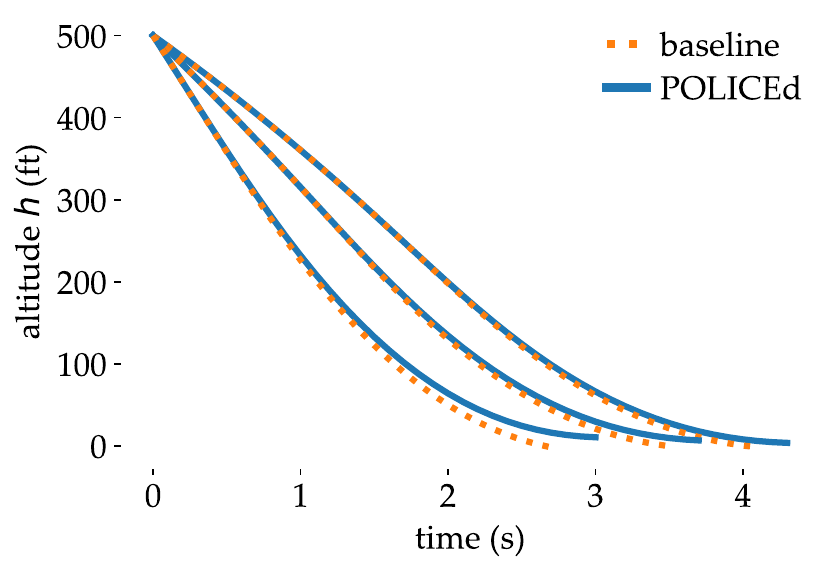}
    \end{subfigure}\hfill
    \begin{subfigure}[t]{0.32\textwidth}
        \includegraphics[scale=0.42]{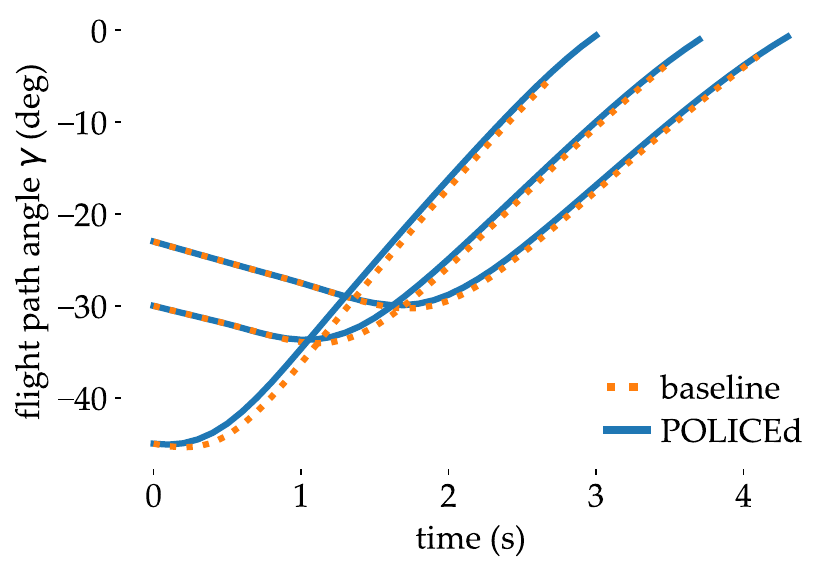}
    \end{subfigure}\hfill
    \begin{subfigure}[t]{0.32\textwidth}
        \includegraphics[scale=0.42]{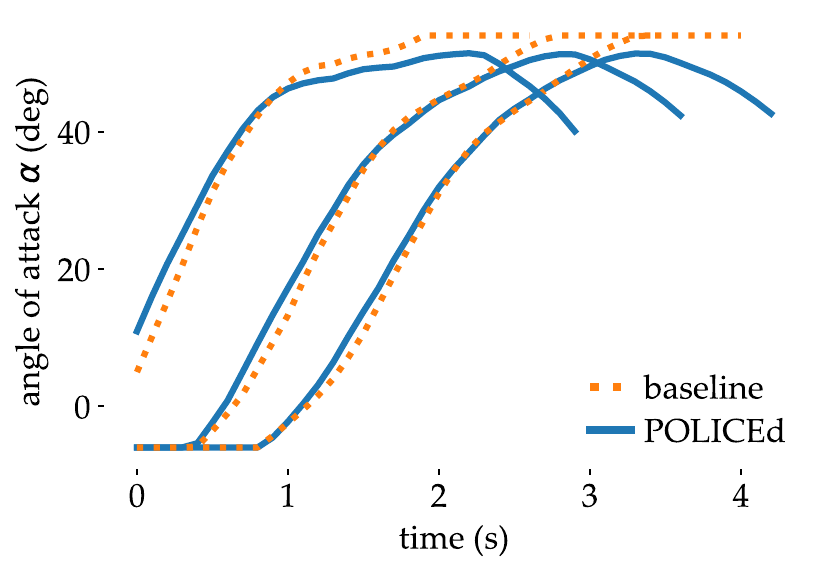}
    \end{subfigure}\hfill
    \caption{Evolution of the altitude $h$, flight path angle $\gamma$ and angle of attack $\alpha$ for three trajectories of the shuttle. All POLICEd trajectories reach level flight $\gamma = 0^\circ$ at touchdown, which also appears as the flattening of the altitude plot at $h = 0$ ft, on the contrary to the baseline. The POLICEd controller also learns to reduce the angle of attack $\alpha$ for landing, while the baseline remain at saturation.}
    \label{fig: shuttle time plots}
\end{figure*}

\subsection{Supporting Lemmata}\label{subsec: lemmata}

Recall that Theorem~\ref{thm: admissible trajectories} guarantees the respect of upper bound $\overline{b}$ only while the trajectory remains in $\mathcal{B}$. We can thus strengthen Theorem~\ref{thm: admissible trajectories} by deriving a lower bound $\underline{b}$ sufficiently low for upper bound $\overline{b}(s)$ never to cross $\underline{b}$ and cause a trajectory to prematurely exit buffer $\mathcal{B}$.
To put it simply, we want $\underline{b} \leq \overline{b}(s)$ for all $s \in \mathcal{B}$.

\begin{lemma}\label{lemma: lower bounds}
    Condition $\underline{b} \leq \overline{b}(s)$ for all $s \in \mathcal{B}$ is equivalent to $s_2^{min} \leq 0$, $s_{2k+1}^{min} \leq -\beta^{2k-1} \dot y_{max}$ and $s_{2k+2}^{min} \leq \beta^{2k} s_{2}^{min}$ for all $k \in [\![1, r/2]\!]$.
\end{lemma}
\begin{proof}
    Let $s \in \mathcal{B}$. Then, $s_2 \in \big[ s_2^{min}, \beta (y_{max} - s_1) \big]$. 
    \begin{equation*}
        s_2^{min} \leq \underset{s}{\min}\, \overline{b}_2(s) = \underset{s}{\min}\, \beta (y_{max} - s_1) = 0,
    \end{equation*}
    since $s_1 \leq y_{max}$. Hence, $s_2^{min} \leq 0$. Similarly,
    \begin{equation*}
        s_3^{min} \leq \underset{s}{\min}\, \overline{b}_3(s) = \underset{s}{\min} -\beta s_2 = -\beta \underset{s}{\max}\, s_2 = -\beta \dot y_{max}.
    \end{equation*}
    For $k \in [\![2, r-2]\!]$, we have
    \begin{align*}
        s_{k+2}^{min} &\leq \underset{s}{\min}\, \overline{b}_{k+2}(s) = \underset{s}{\min} -\beta s_{k+1} = -\beta \underset{s}{\max}\, s_{k+1} \\
        &\leq -\beta \underset{s}{\max}\, \overline{b}_{k+1}(s) = -\beta \underset{s}{\max} -\beta s_k = \beta^2 \underset{s}{\min}\, s_k \\
        &\leq \beta^2 s_k^{min}.
    \end{align*}
    Applying this inequality recursively leads to $s_{2k+1}^{min} \leq -\beta^{2k-1} s_3^{min} \leq -\beta^{2k-1} \dot y_{max}$ and $s_{2k+2}^{min} \leq \beta^{2k} s_{2}^{min}$ for all $k \in [\![1, r/2]\!]$.
\end{proof}

Armed with Lemma~\ref{lemma: lower bounds}, we can now calculate the minimal number of vertices of buffer $\mathcal{B}$, which corresponds to the case where $\underline{b} = \underset{s}{\min}\, \overline{b}(s)$.
We first need to introduce the Fibonacci sequence as $F_{n+2} = F_{n+1} + F_n$ for all $n \in \mathbb{N}$ with $F_0 = 0$ and $F_1 = 1$ \cite{Fibonacci}.
We recall that $r$ is the relative degree of output $y$ of \eqref{eq: constraint} with respect to dynamics~\eqref{eq: nonlinear dynamics}. We denote the cardinality of a set $\mathcal{Z}$ as $\big|\mathcal{Z}\big|$.

\begin{lemma}\label{lemma: number vertices}
    If $\underline{b} = \underset{s}{\min}\, \overline{b}(s)$, then the number of vertices of $\mathcal{B}$ is $F_{r+2} \times \big|\mathcal{V}(\mathcal{P})\big|$.
\end{lemma}
\begin{proof}
    We introduce $\mathcal{B}^r := \big\{ s_{1:r} \in [\underline{b},\ \overline{b}(s)] : s \in \mathcal{B}\big\}$.
    Then, buffer $\mathcal{B}$ of \eqref{eq: buffer} can be written as the Cartesian product $\mathcal{B} = \mathcal{B}^r \times \mathcal{P}$. Since $\big|\mathcal{V}(\mathcal{P})\big|$ denotes the number of vertices of polytope $\mathcal{P}$, we only need to show that the number of vertices of $\mathcal{B}^r$ is equal to $F_{r+2}$.

    Following Lemma~\ref{lemma: lower bounds},  $\underline{b} = \underset{s}{\min}\, \overline{b}(s)$ yields, $s_{2k}^{min} = 0$ and $s_{2k+1}^{min} = -\beta^{2k-1}\dot y_{max}$ for all $k \in [\![1, r/2]\!]$.
    Since upper bound $\overline{b}_k(s)$ depends on the value of $s_{k-1}$, to enumerate the vertices of $\mathcal{B}^r$ we need to first list the possibilities for $s_1$, then for $s_2$, and so on.
    At dimension $k \in [\![2, r]\!]$, we have $s_k \in [s_k^{min}, -\beta s_{k-1}]$. If $s_{k-1} = s_{k-1}^{min}$, then $s_k \in [s_k^{min}, s_k^{max}]$. If $s_{k-1} = s_{k-1}^{max}$, then $s_k^{max} = -\beta s_{k-1} = s_k^{min}$ by assumption $\underline{b} = \min_s\, \overline{b}(s)$.
    
    To enumerate the number of vertices of $\mathcal{B}^r$, we build a tree listing all the possibilities where each level correspond to a dimension as shown in Fig.~\ref{fig: tree}.

    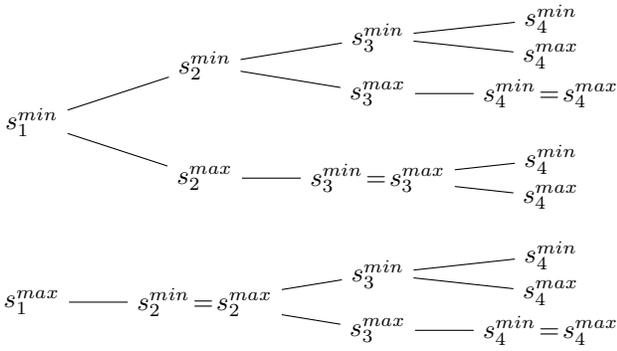
\begin{figure}[htbp!]
        \centering
        \begin{tikzpicture}
            \node at (0, 2.4) {$s_1^{min}$} [grow'= right, level distance=23mm, level/.style={sibling distance=15mm/#1}]
                child {node {$s_2^{min}$}
                    child {node {$s_3^{min}$}
                        child {node {$s_4^{min}$}}
                        child {node {$s_4^{max}$}}
                        }
                    child {node {$s_3^{max}$}
                        child {node {$s_4^{min}\! =\! s_4^{max}$}}
                    }
                }
                child {node {$s_2^{max}$}
                    child {node {$s_3^{min}\! =\! s_3^{max}$}
                        child {node {$s_4^{min}$}}
                        child {node {$s_4^{max}$}}
                  }
                };            
            \node at (0, 0) {$s_1^{max}$} [grow'= right, level distance=23mm, level/.style={sibling distance=15mm/#1}]
                child {node {$s_2^{min}\! =\! s_2^{max}$}
                    child {node {$s_3^{min}$}
                        child {node {$s_4^{min}$}}
                        child {node {$s_4^{max}$}}
                    }
                    child {node {$s_3^{max}$}
                        child {node {$s_4^{min}\! =\! s_4^{max}$}}
                    }
                };
        \end{tikzpicture}%
        \caption{Tree of all the vertex combinations for $\mathcal{B}_r$. A node $s_k^{min}$ allows $s_{k+1}$ to take value in its whole range $[s_{k+1}^{min}, s_{k+1}^{max}]$ and thus yields two nodes $s_{k+1}^{min}$ and $s_{k+1}^{max}$. However, a node $s_k^{max}$ causes $\overline{b}_k = \underline{b}_k$ and thus yields a single node $s_{k+1}^{min} = s_{k+1}^{max}$. Finally, a node $s_{k}^{min} = s_{k}^{max}$ forces $s_k = s_k^{min}$ and thus $s_{k+1}$ can use its whole range, which yields two nodes $s_{k+1}^{min}$ and $s_{k+1}^{max}$.}
        \label{fig: tree}
    \end{figure}

    We count the nodes of the tree. At level $k \in [\![1,r]\!]$ we define the number of nodes where $s_k = s_k^{min}$ as $n^{min}(k)$, the number of nodes where $s_k = s_k^{max}$ as $n^{max}(k)$, the number of nodes where $s_k = s_k^{min} = s_k^{max}$ as $n^{nx}(k)$, and the total number of nodes as $n(k) := n^{min}(k) + n^{max}(k) + n^{nx}(k)$.
    
    A node $s_k^{min}$ allows $s_{k+1}$ to take value in its whole range $[s_{k+1}^{min}, s_{k+1}^{max}]$ and thus yields two nodes $s_{k+1}^{min}$ and $s_{k+1}^{max}$. However, a node $s_k^{max}$ causes $\overline{b}_k = \underline{b}_k$ and thus yields a single node $s_{k+1}^{min} = s_{k+1}^{max}$. Finally, a node $s_{k}^{min} = s_{k}^{max}$ forces $s_k = s_k^{min}$ and thus $s_{k+1}$ can use its whole range, which again yields two nodes $s_{k+1}^{min}$ and $s_{k+1}^{max}$.
    
    Thus, each node $s_{k+1}^{min} = s_{k+1}^{max}$ is created by a node $s_{k}^{max}$, so that $n^{nx}(k+1) = n^{max}(k)$. On the other hand, there is a $s_{k+1}^{min}$ and $s_{k+1}^{max}$ node for each $s_{k}^{min}$ and each $s_{k}^{min} = s_{k}^{max}$ nodes, thus $n^{min}(k+1) = n^{max}(k+1) = n^{min}(k) + n^{nx}(k)$.
    
    Then, for $k \in [\![3, r]\!]$, we have
    \begin{align*}
        n(k+1) &= n^{min}(k+1) + n^{max}(k+1) + n^{nx}(k+1) \\
        &\hspace{-1.3mm}= 2n^{min}(k) + n^{max}(k) + 2n^{nx}(k) \\
        &\hspace{-2.6mm}= n(k) + n^{min}(k) + n^{nx}(k) \\
        &\hspace{-4mm} = n(k) + n^{min}(k-1) + n^{nx}(k-1) + n^{max}(k-1) \\
        &\hspace{-5mm}= n(k) + n(k-1).
    \end{align*}
    Therefore, $n$ verifies the recursion relation of the Fibonacci sequence \cite{Fibonacci}. Additionally, Fig.~\ref{fig: tree} shows that $n(1) = 2 = F_3$ and $n(2) = 3 = F_4$. Thus, $n(k) = F_{k+2}$.
\end{proof}

\bibliographystyle{IEEEtran}
\bibliography{references}

\end{document}